\documentclass[twocolumn]{article}

\usepackage{color}
\usepackage{boxedminipage}
\usepackage{xspace}
\usepackage{url}
\usepackage{tikz}
\usepackage{subcaption}
\usepackage{amsthm}

\usepackage{floatrow}
\usepackage{listings}
\definecolor{listinggray}{gray}{0.9}
\usepackage{enumitem}

\newcommand{\ignore}[1]{{}}
\newcommand{\eg}{\textit{e.g.,}\xspace}
\newcommand{\ie}{\textit{i.e.,}\xspace}

\newtheorem{lemma}{Lemma}
\newtheorem{corollary}{Corollary}
\newtheorem{fact}{Fact}

\title{\bf POSH: Paris OpenSHMEM\\
A High-Performance OpenSHMEM Implementation for Shared Memory Systems}

\author{Camille Coti
\thanks{Some experiments presented in this paper were carried out using the Grid'5000
experimental testbed, being developed under the INRIA ALADDIN development action
with support from CNRS, RENATER and several Universities as well as other
funding bodies (see \protect\url{https://www.grid5000.fr}).}\\
\normalsize LIPN, CNRS-UMR7030,\\
\normalsize Universit\'{e} Paris 13, F-93430 Villetaneuse, France\\
\texttt{camille.coti@lipn.univ-paris13.fr}
}

\date{}

\begin{document}

\maketitle

\begin{abstract}
In this paper we present the design and implementation of POSH, an 
Open-Source implementation of the OpenSHMEM standard. We present a 
model for its communications, and prove some properties on the memory 
model defined in the OpenSHMEM specification. We present some 
performance measurements of the communication library featured by POSH 
and compare them with an existing one-sided communication library. 
POSH can be downloaded from \url{http://www.lipn.fr/~coti/POSH}. 
\end{abstract}

\section{Introduction}
\label{sec:intro}

The drive toward many-core architectures has been tremendous during the last
decade. Along with this trend, the community has been searching, investigating
and looking for programming models that provide both control on the data
locality and flexibility of the data handling. 

SHMEM was introduced by Cray~\cite{CraySHMEM} in 1994, followed shortly later by
SGI~\cite{SGISHMEM}. In an effort to provide a homogeneous, portable standard for
the language, the OpenSHMEM consortium released a specification for the
application programming interface~\cite{openshmem}. The final version of
OpenSHMEM 1.0 was released in January 2012. 

The OpenSHMEM standard is a programming paradigm for parallel applications that
uses single-sided communications. It opens gates for exciting research in
distributed computing on this particular communication model. 

This paper presents Paris OpenSHMEM (POSH), which is a portable, open-source
implementation of OpenSHMEM. It uses a high-performance communication engine on
shared memory based on the Boost library~\cite{BoostIPC}, and benefits from the
template engine provided by current C++ compilers.

This paper describes the implementation choices and the algorithms
that have been used in POSH in order to fit with the memory model and the
communication model while obtaining good performance and maintaining
portability. 

This report is organized as follows; section \ref{sec:related} gives a short
overview of the related literature about parallel programming paradigms and
distributed algorithms on shared memory and one-sided communication
models. Section \ref{sec:model} gives details about the memory model and the
communication model which are considered here. Section \ref{sec:implem} describes
the implementation choices that were made in POSH. Section \ref{sec:perf}
presents some performance results that were obtained by the current version of
POSH. Last, section \ref{sec:conclu} concludes the report and states some open
issues and future works that will be conducted on POSH. 
\section{Related works}
\label{sec:related}

Traditionally, distributed systems are divided into two categories of models for
their communications: those that communicate by sending and receiving messages
(\ie message-passing systems) and those that communicate using registers of
shared memory where messages are written and read from (\ie shared memory
systems) \cite{Tel94, D00}. 

Along with the massive adoption of the many-core hardware architecture,
researchers and engineers have tried to find the most efficient programming
paradigm for such systems. The idea is to take advantage of the fact that
processing units (processes or threads) have access to a common memory: those
are {\it shared memory} systems. Unix IPC V5 and posix threads are the most
basic programming tools for that. OpenMP \cite{openmp} provides an easy-to-use
programming interface and lets the programmer write programs that look very
similar to sequential ones, and the parallelization is made by the
compiler. Therefore, the compiler is in charge with data decomposition and
accesses. Several data locality policies have been implemented to try to make
the best possible guess about where it must be put to be as efficient as
possible \cite{weng2002implementing}. OpenMP performs well on regular patterns,
where the data locality can be guessed quite accurately by the compiler. Cilk
\cite{blumofe1995cilk} and TBB \cite{reinders2010intel} can also be cited as
programming techniques for shared-memory systems.

MPI \cite{Forum94, MPI-2} has imposed itself as the {\it de facto} programming
standard for distributed-memory parallel systems. It is highly portable, and
implementations are available for a broad range of platforms. MPICH
\cite{mpich2} and Open~MPI \cite{openmpi} must be cited among the most widely
used open-source implementations. It can be used on top of most local-area
communication networks, and of course most MPI implementations provide an
implementation on top of shared memory. MPI is often referred to as ``the assembly 
language of parallel computing": the programmer has total control of the data
locality, however all the data management must be implemented by hand by the
programmer. Moreover, it is highly synchronous: even though specific one-sided
communications have been introducted in the MPI2 standard \cite{MPI-2}, the
sender and the receiver must be in matching communication routines for a
communication to be performed. 

Hence, there exists two opposing trends in parallel programming techniques:
programming easiness versus data locality mastering. A third direction exists
and is becoming pertinent with many-core architectures: making data 
locality mastering easier and more flexible for the programmer. UPC can be cited
as an example of programming technique that is part of that third category
\cite{UPC}. It provides compiler-assisted automatic loops, automatic data
repartition in (potentially distributed) shared memory, and a set of one-sided
communications. 

One-sided communications are natural on shared memory systems, and more flexible
than two-sided communications in a sense that they do not require that both of
the processes involved in the communication (origin and destination of the data)
must be in matching communication routines. However, they require a careful
programming technique to maintain the consistency of the shared memory and avoid
race conditions \cite{BC12}.

SHMEM was introduced by Cray \cite{CraySHMEM} as part of its programming
toolsuite with the Cray T3 series, and SGI created its own dialecte of
SHMEM \cite{SGISHMEM}. 

Some implementations also exist for high-performance RDMA networks: Portals has
been working on a specific support for OpenSHMEM by their communication
library~\cite{PortalsOpenSHMEM11}. Some other implementations are built on
top of MPI implementations over RDMA networks, such as~\cite{Brightwell04anew}
for Quadrics networks or~\cite{Liu03highperformance} over InfiniBand networks.

In this paper, we propose to use a shared memory communication engine based on
the Boost.Interprocess library~\cite{BoostIPC}, which is itself using the POSIX
{\tt shm} API. 

\ignore{ implemsn Ron Brightwell sur MPI
UCCS http://www.csm.ornl.gov/workshops/openshmem2013/documents/DesigningAHighPerformanceOpenSHMEMImplementation.pdf
}
\section{Memory model}
\label{sec:model}

\begin{figure*}[ht]
  \centering
  \begin{tikzpicture}[scale=.45]

    \foreach \x in { 0, 1, 2 } {
      \path[draw,thick] ( \x*5, 0 ) rectangle ( \x*5+2, -4 );
      \path[draw,thick] ( \x*5, 0 ) rectangle ( \x*5+2, -2 );
      \node at ( \x*5+1, -1 ) {\textbf{P\x}};
      \path[draw,thick,fill=gray!40] ( \x*5, -4 ) rectangle ( \x*5+2, -8 );
      \path[draw,thick,dashed] ( \x*5+.2, -5.9 ) rectangle ( \x*5+1.8, -7.8 );
      \path[draw,thick,<-] ( \x*5+1, -5.9 ) -- ( \x*5+1, -5.4 );

      \path[draw,thick,fill=black] ( \x*5+.2, -4.2 ) rectangle ( \x*5+1, -4.4 );
      \path[draw,thick,<-] ( \x*5+.6, -4.4 ) -- ( \x*5+.6, -5 );

      \path[draw,thick,fill=black] ( \x*5+.4, -6.2 ) rectangle ( \x*5+1.6, -6.4 );
      \path[draw,thick,<-] ( \x*5+1.4, -6.4 ) -- ( \x*5+1.4, -7.6 );

      \path[draw,thick,fill=black] ( \x*5+.4, -6.8 ) rectangle ( \x*5+1.2, -7 );
      \path[draw,thick,<-] ( \x*5+.7, -7 ) -- ( \x*5+.7, -7.6 );
    }
    
    \path[draw,thick] ( .6, -5 ) -- ( 13, -5 );
    \path[draw,thick] ( .7, -7.6 ) -- ( 13, -7.6 );
    \path[draw,thick] ( 11, -5.4 ) -- ( -1, -5.4 );

    \path[draw,thick] ( .2, -2.2 ) rectangle ( .5, -2.4 );
    \path[draw,thick] ( .7, -2.2 ) rectangle ( 1.2, -2.4 );
    \path[draw,thick] ( .7, -3 ) rectangle ( 1.2, -3.2 );

    \path[draw,thick] ( 5.2, -2.2 ) rectangle ( 6.2, -2.4 );
    \path[draw,thick] ( 5.4, -2.6 ) rectangle ( 6, -2.8 );

    \path[draw,thick] ( 10.4, -2.2 ) rectangle ( 10.6, -2.4 );
    \path[draw,thick] ( 10.2, -2.6 ) rectangle ( 11, -2.8 );
    \path[draw,thick] ( 11, -3 ) rectangle ( 11.8, -3.2 );

    \node at ( -3, -2.4 ) {Private};
    \node at ( -3, -3.4 ) {memory};

    \node at ( -3, -5 ) {Symmetric};
    \node at ( -3, -6 ) {heap};

    \node at ( 15, -4.5 ) {Static global};
    \node at ( 15, -5.2 ) {objects};

    \node at ( 15, -7 ) {Symmetric};
    \node at ( 15, -7.7 ) {objects};

  \end{tikzpicture}
  \caption{\label{fig:symmetricheap}Memory organization with global static
    objects and data in the symmetric heap. Static objects are remotely
    accessible, dynamic objects are located in the symmetric heap.} 
\end{figure*}
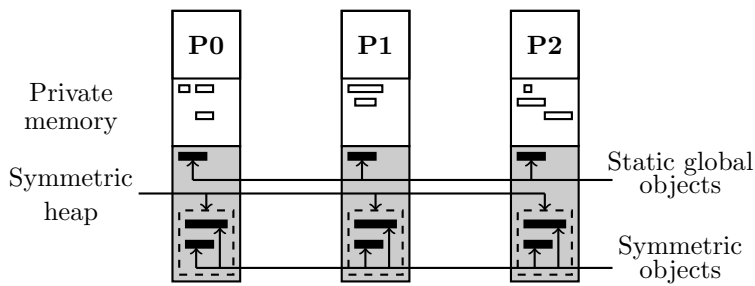

OpenSHMEM considers a memory model where every process of the parallel
application owns a local bank of memory which is split into two parts:
\begin{itemize}[noitemsep]
\item Its \emph{private memory}, which is accessible by itself only; no other
  process can access this area of memory.
\item Its \emph{public memory}, that can be accessed by any process of the
  parallel application, in read/write mode.
\end{itemize}

This memory organization is represented in figure \ref{fig:symmetricheap}. Each
process owns its private memory (white rectangles) as well as an area of public
memory (gray rectangles).

\subsection{Symmetric objects}

The public memory of each process is defined as a \emph{symmetric heap}. This
notion of symmetry is important because it is a necessary condition for some
helpful memory-management properties in OpenSHMEM (see section
\ref{sec:implem:remoteaccess} and \ref{sec:implem:collective:temp}). It means
that for any object which is stored in the symmetric heap of a process, there
exists an object of the same type and size and the same address, in the
symmetric heap of all the other processes of the parallel application. 

\emph{Dynamically-allocated variables}, \ie variables that are allocated
explicitely at run-time, are placed in the symmetric heap. OpenSHMEM provides
some functions that allocate and deallocate memory space dynamically in the
symmetric heap.

Another kind of data is put into processes' public memory: global, static
variables (in C/C++). 

All the data that is placed in the processes' symmetric heap and all the global,
static variables are remotely accessible by all the other processes. These two
kinds of variables are represented in figure \ref{fig:symmetricheap}. The gray
areas prepresent the public memory of each process; composed with global static
variables and the symmetric heap. The white areas represent the private memory
of each process.

\subsection{One-sided communications}

\begin{figure*}[ht]
  \centering
  \begin{tikzpicture}[scale=.5]
    
    \foreach \x in { 0, 1, 2 } {
      \path[draw,thick] ( \x*5, 0 ) rectangle ( \x*5+2, -6 );
      \path[draw,thick] ( \x*5, 0 ) rectangle ( \x*5+2, -2 );
      \node at ( \x*5+1, -1 ) {\textbf{P\x}};
      \path[draw,thick,fill=gray!40] ( \x*5, -6 ) rectangle ( \x*5+2, -10 );
    }
    
    \node at ( 14.5, -3 ) {Private};
    \node at ( 14.5, -4 ) {Address};
    \node at ( 14.5, -5 ) {Space};
    
    \node at ( 14.5, -7 ) {Public};
    \node at ( 14.5, -8 ) {Address};
    \node at ( 14.5, -9 ) {Space};
    
    \path[draw,thick,fill=black] ( 10.8, -3.2 ) rectangle ( 11.2, -2.8 );
    \path[draw,thick,fill=black] ( 6.2, -6.4 ) rectangle ( 5.8, -6.8 );
    \draw[ ->, thick] ( 6, -6.6 ) -- ( 10.8, -3.2 );
    \node at ( 8.5, -3 ) {\small Remote};
    \node at ( 8.5, -3.8 ) {\small get};
    
    \path[draw,thick,fill=black] ( .8, -4.2 ) rectangle ( 1.2, -3.8 );
    \path[draw,thick,fill=black] ( 6.2, -7.4 ) rectangle ( 5.8, -7.8 );
    \draw[ ->, thick] ( 1, -4 ) -- ( 5.8, -7.4 );
    \node at ( 3.5, -4 ) {\small Remote};
    \node at ( 3.5, -4.8 ) {\small put};

    \path[draw,thick,fill=black] ( .8, -9.2 ) rectangle ( 1.2, -8.8 );
    \path[draw,thick,fill=black] ( 6.2, -9.2 ) rectangle ( 5.8, -8.8 );
    \draw[ ->, thick] ( 1, -9 ) -- ( 5.8, -9 );
    \node at ( 3.5, -7.7 ) {\small Remote};
    \node at ( 3.5, -8.5 ) {\small put};
  \end{tikzpicture}
  \caption{\label{fig:putget}Put and get communications between three processes.}
\end{figure*}
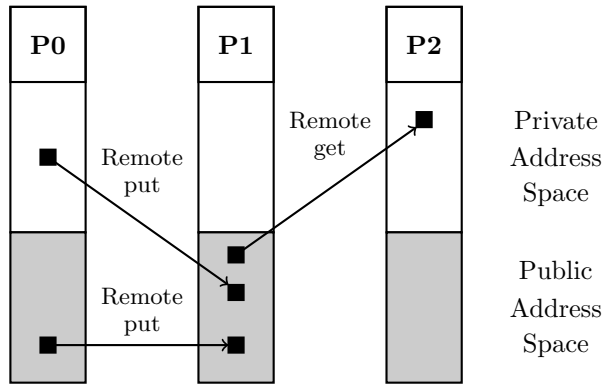

Point-to-point communications in OpenSHMEM are \emph{one-sided}: a process can
reach the memory of another process without the latter knowing it. It is very
convenient at first glance, because no synchronization between the two processes
is necessary like with two-sided communications. However, such programs must be
programmed very carefully in order to maintain memory consistency and avoid
potential bugs such as race conditions. 

Point-to-point communications in OpenSHMEM are based on two primitives:
\texttt{put} and \texttt{get}.
\begin{itemize}[noitemsep]
\item A \texttt{put} operation consists in writing some data at a specific
  address of remote process's public memory. 
\item A \texttt{get} operation consists in reading some data, or fetching it,
  from a specific address of a remote process's public memory. 
\end{itemize}

Data movements are made between the public memory of the local process and the
private memory of the remote process: a process reads a remote value and stores
it in its own private memory, and it writes the value of a variable located in
its own private memory into the public memory of a remote process.  
That memory model and one-sided communications performed on this model have been
described more thoroughly in \cite{BC11, BC12}.

These operations between the public and private memory areas of three processes
of a parallel application are represented in figure \ref{fig:putget}.

\section{Implementation details}
\label{sec:implem}

In this section, we describe the implementation details and the design choices
that have been made in POSH.

\subsection{Shared memory communication engine}
\label{sec:implem:commengine}

POSH relies on Boost's library for inter-process communications
\texttt{Boost.Interprocess}. In particular, it is using the
\texttt{managed\_shared\_memory} class.

Basically, each process's shared heap is an instance of
\texttt{managed\_shared\_memory}. Data is put into that heap by an allocation
method provided by this class followed by a memory copy. Locks and all the
manipluation functions related to the shared memory segment are also provided
by Boost.

\subsubsection{Memory management}
\label{sec:implem:memmanagement}

\paragraph{Allocation and deallocation}

Memory can be allocated in the symmetric heap with the \texttt{shmalloc}
function. Internally, that function calls the \texttt{allocate} function of the
\texttt{managed\_shared\_memory} class on the process's shared heap. That class
also provides an \texttt{allocate\_aligned} routine which is called by the
\texttt{shmemalign} routine. Memory is freed by \texttt{shfree} by using a call
to the function \texttt{deallocate} provided by
\texttt{managed\_shared\_memory}. 

\paragraph{Remote vs local address of the data}

These three SHMEM functions are defined as \emph{symmetric} functions: all the
processes must call them at the same time. They are required by the OpenSHMEM
standard to perform a global synchronization barrier before returning. As a consequence,
if all the memory allocations and deallocations have been made in a symmetric
way, a given chunk of data will have the same local address in the memory of all
the processes. Hence, we can access a chunk of data on process A using the
address it has on process B. That property is extremely useful for remote data
accesses. 

\subsubsection{Access to another process's shared heap}
\label{sec:implem:remoteaccess}

\begin{fact}\label{lemma:memorysymmetry}
If all the processing elements are running on the same architecture, the offset
between the beginning of a symmetric heap and a symmetric object which is
contained by this heap is the same on each processing element. 
\end{fact}

\begin{proof}
Thanks to the property of symmetry between the shared heaps, data is put into
the symmetric heaps at the same \emph{local} address across processes. For
instance, if a variable is located in the shared heap of a process at address
$X$, the same variable (possibly with a different value) will be located at the
same address $X$ in the symmetric heap of all the other processes.

That property is enforced by the fact that, as stated in the OpenSHMEM standard,
memory allocations which are performed in the symmetric heaps end by a call to a
global synchronization barrier. As a consequence, all the processing elements
\emph{must} allocate space. If they do not allocate the same space, then this is a
mistake made by the programmer and the behavior of the program is undefined
(paragraph 6.4 of the OpenSHMEM standard).
\end{proof}

\begin{corollary}
As a consequence of fact \ref{lemma:memorysymmetry}, each processing element
can compute the address of a variable located in another processing element's
symmetric heap by using the following formula:
\begin{equation}
  addr_{remote} = heap_{remote} + ( addr_{local} - heap_{local} )
\end{equation}
Where $addr_{remote}$ and $addr_{local}$ respectively denote the address of
symmetric objects that are located in the symmetric heap of a remote processing
element and in the local symmetric heap, and $heap_{remote}$ and $heap_{local}$
respectively denote the address at which the other processing element's heap and the
local heap are mapped in the local memory.
\end{corollary}

When a given process wants to access another process's symmetric heap, it
proceeds as follows:

\begin{itemize}[noitemsep]
\item Build the remote symmetric heap's name, based on its rank;
\item Make sure the remote symmetric heap exists. If it does not exist yet, we
  wait a little bit and try again;
\item Build a local object that contains a reference to the shared memory and
  maps it into the local memory;
\item Then find out what the address of the remote object is in the local
  memory. We use a little trick here. Boost provides the notion of \emph{handle}
  to locate a chunk of memory in a shared memory segment. This handle is
  relative to a shared memory segment. So we get the handle to the chunk of data
  which is in \emph{our own} symmetric heap, and then we get the address that
  corresponds to this handle in the \emph{remote} symmetric heap. This trick is
  allowed by the property of symmetry between all the symmetric heaps, which is
  explained in section \ref{sec:implem:memmanagement}. 
\end{itemize}

That being done, we have direct access to the remote symmetric heap. In
particular, we can copy data from and to this remote symmetric
heap.

However, building the remote heap's name and the corresponding shared object is
quite expensive in terms of object crations (and destructions at the end of this
process). As a consequence, they are all created at startup-time and cached in
a local structure (a table). This operation in itself is quite inexpensive
thanks to move mechanisms that are featured by C++11. Hence, when a process
needs to access another process's heap, it simply looks into this local table
for the reference to this segment of shared memory and accesses it.

\subsection{Symmetric static data}
\label{sec:implem:static}

The memory model specifies that global, static variables are made accessible for
other processes. In practice, these variables are placed in the BSS segment if
they are not initialized at compile-time and in the data segment if they are
initialized. Unfortunately, there is no simple way to make these areas of memory
accessible for other processes.

Therefore, POSH uses a small trick: we put them into the symmetric heap at the
very beginning of the execution of the program, before anything else is done. 

A specific program, called the \emph{pre-parser}, parses the source code and
searches for global variables that are declared as static. It finds out how they
must be allocated (size, etc) and generates the appropriate
allocation/deallocation code lines. 

When the OpenSHMEM library is initialized (\ie when the \texttt{start\_pes}
routine is called), it dumps the allocation code into the source code. When the
program exits (\ie when the keyword \texttt{return} is found in the main
function), the deallocation code lines are inserted before each \texttt{return}
keyword. 

\subsection{Datatype-specific routines}
\label{sec:implem:templates}

OpenSHMEM defines a function for each data type. For example, fetching a single
variable can be done by:
\begin{itemize}[noitemsep]
\item \texttt{short shmem\_short\_g( short *addr, int pe )} for variables of type \texttt{short}
\item \texttt{int shmem\_int\_g( int *addr, int pe )} for variables of type \texttt{int}
\item \texttt{long  shmem\_long\_g( long *addr, int pe )} for variables of type \texttt{long}
\item \texttt{float shmem\_float\_g( float *addr, int pe )} for variables of type \texttt{float}
\item \texttt{double shmem\_double\_g( double *addr, int pe )} for variables of type \texttt{double}
\item \texttt{long long shmem\_longlong\_g( long long *addr, int pe )} for variables of
  type \texttt{long long}
\item \texttt{long double shmem\_longdouble\_g( long double *addr, int pe )} for variables
  of type \texttt{long double}
\end{itemize}

A large part of this code can be factorized by using an extremely powerful
feature of the C++ language: templates. The corresponding code is written only
once, and then the template engine instanciates one function for each data
type. Hence, only one function needs to be written.

In the aforementioned example, only one function was written:

\begin{center}
  \begin{minipage}{.92\linewidth}
    \begin{lstlisting}
template<class T> T shmem_template_g( T* addr, int pe );
    \end{lstlisting}
  \end{minipage}
\end{center}

That function is called by each of the OpenSHMEM \texttt{shmem\_*\_...}
functions. Each call is actually a call to the compiler-generated function that
uses the adequate data type. That function is generated at compile-time, not at
run-time: consequently, calling that function is just as fast as if it had been
written manually. 

\subsection{Peer-to-peer communications}
\label{sec:implem:p2p}

Peer-to-peer communications are using memory copies between local and
shared buffers. As a consequence, memory copy is a highly critical
matter of POSH. Several implementations of \texttt{memcpy} are featured by POSH
in order to make use of low-level hardware capabilities such as MMX, MMX2, SSE
or SSE2 instruction sets, or the default \texttt{memcpy} provided by the
kernel. 

One of these implementations is activated by using a compiler directive. In
order to minimize the number of conditional branches, selecting one
particular implementation is made at compile-time rather than at run-time.

A comparison of the performance obtained by these different implementations is
presented in section \ref{sec:perf}.

\subsection{Collective communications}
\label{sec:implem:collective}

Collective communications rely on point-to-point communications that
perform the actual inter-process data movements. Two options are
available for these point-to-point communications:

\begin{itemize}[noitemsep]
\item \emph{Put-based} communications push the data into the
  next processes; 

\item \emph{Get-based} communications pull the data from other
  processes.
\end{itemize}

\subsubsection{Data structure for collective communication}

Each process holds a data structure in their shared heap (hence, other processes
can access it). This data structure contains information about the ongoing
collective operation: 

\begin{itemize}[noitemsep]
\item A pointer to the buffer that contains the data that is moved by
  the collective operation;
\item A counter, that counts how many remote processes have accessed
  the local data;
\item A type, that keeps what collective operation is underway;
\item A boolean that specifies whether or not the collective
  communication is already in progress;
\item In debug and in safe mode we can keep the size of the data
  buffer, in order to check that the allocated buffer has the same
  size as the data we are trying to put into it.
\end{itemize}

This data structure is intialized during the initalization of the OpenSHMEM
library, after the initalization of the symmetric heap. It is reset at the end
of each collective communication, in order to make sure the place is "clean" for the
next collective communication.

\subsubsection{Progress of a collective operation}

The communication model used by OpenSHMEM and its point-to-point
communication engine is particular in a sense that it is using
\emph{one-sided} operations. As a consequence, a process can access in
read or write mode another process's memory without the knowledge of
the latter process. One consequence of this fact is that a process can
be involved in a collective communication without having actually
entered the call to the corresponding routine yet. 

Hence, if a process $A$ must access the symmetric heap of a process $B$,
the former process must check whether or not the latter has entered
the collective communication yet. A boolean variable is included in the
collective data structure for this purpose.

If the remote process has not entered the collective communication
yet, its collective data structure must be initialized
remotely. 

If some data must be put into a remote process that has yet to
initialize its collective data structure, we only copy the
pointer to the shared source buffer. The actual memory allocation will
be made later. However, only temporary memory allocations are made within
collective operations. Buffers that are used as parameters of a collective
operation (source, target and work arrays) must be allocated before the call to
this operation. Since memory allocations end by a global barrier, no processing
element can enter a collective operation if not all of them have finished
their symmetric memory allocations.

When a process enters a collective operation, it checks whether the
operation is already underway, \ie whether its collective
data structure has already been modified by a remote process. If so,
we need to make the actual memory allocation for the local data and
copy what has already been put somewhere in another shared memory
area. 

A process exits the collective communication as soon as its
participation to the communication is over. Hence, no other process
will access its collective data structure. It can
therefore be reset.

\subsubsection{Temporary allocations in the shared heap}
\label{sec:implem:collective:temp}

With some collective communication algorithms, it can be necessary to allocate
some temporary space in the shared heap of a given processing element. However,
if we allocate some memory in one heap only, we break the important symmetry
assumption made in section \ref{sec:implem:memmanagement}. Nevertheless, we will see
here that actually, they have no impact in the symmetry of the shared heaps
outside of the affected collective operation.

\begin{lemma}\label{lemma:memorysymmetrycollective}
Non-symmetric, temporary memory allocations in the heap of a subset of the
processing elements that are performed during collective operations do not break
the symmetry of the heaps outside of the concerned collective operation.
\end{lemma}

\begin{proof}
Semantically, collective operations are symmetric, in a sense that all the
concerned processing elements must take part of them. As a consequence, if all
the heaps are symmetric before they enter the collective operation and if there
is no difference between the state of each heap at the beginning and at the end
of the collective operation, hence, the symmetry is not broken. 

Buffers (source, target and work arrays) must be allocated by a call to
\texttt{shmalloc} before calling the collective operation. Since
\texttt{shmalloc} ends by a call to a global barrier, no processing element can
enter the collective operation before all the processing elements have performed
the last memory allocation before this collective call.

According to the sematics of collective operations, at any given moment of the
progress of a collective operation, the relative state of each pair of processes
$A$ and $B$ is one of the following:
\begin{itemize}[noitemsep]
\item Neither $A$ nor $B$ have entered the collective operation yet, or both of
  them have left it: in this case, $A$ and $B$ do not interact with
  each other's symmetric heaps, and memory allocations are performed
  by symmetric allocations, as described in section
  \ref{sec:implem:memmanagement}.

\item Both processes $A$ and $B$ are inside the call to the collective
  routine and can interfere with each other's symmetric heap: anything
  can happen in the symmetric heaps within the progress of the
  collective operation, as long as non-symmetric operations are
  cleaned-up when each process leaves the collective operation. 

\item A process $A$ can push some data into the memory of another process $B$
  that has not entered the collective operation yet. In this case, $B$ is
  unknowingly taking part of the collective operation, but
  its shared heap can potentially be modified by other processes. However, as
  stated above, $B$ cannot perform any symmetric memory allocation between that
  moment and the moment when it enters the collective operation. As a consequence,
  the non-symmetric allocations that can potentially be made in $B$'s shared heap
  by other processing elements while the latter are taking part of the collective
  operation do not break the symmetry as long as they are reverted (\ie freed)
  before $B$ leaves the collective operation.

\item A process $A$ still part of the collective operation whereas another
  process $B$ is done with its participation to the operation and therefore, $B$
  has left the call to the collective routine, which means that its
  participation to the collective operation is over. As a consequence,
  other processes have no reason to modify its shared heap. 
\end{itemize}

\end{proof}

\subsubsection{Switching between algorithms}

In order to reduce the number of conditional branches, collective
communication algorithms are chosen at compile-time. The choice is
made using compiler variables and conditions. A default choice is
provided if no option is passed to the compiler, and a warning is
displayed. 

\subsubsection{Run-time error checking}

Collective communications are an important source of bugs and errors
in parallel programs. When compiled in safe mode, the OpenSHMEM
library provides some run-time error checking facilities. 

For instance, it can check whether the size of the available buffer is
equal to the size of the data that is about to be pushed into or
pulled from a shared collective data structure. It can also make sure
that the collective data structures of the local and the remote
processes are performing the same type of collective operation. 

\subsection{Locks and atomic operations}

Boost provides \emph{named mutexes} for locking purpose. These locks are
interesting, because they can be specific for a given symmetric heap. Each
process uses the same given name for a given chunk of data on a given symmetric
heap. Using a mutex that locally has the same name as all the other local
mutexes, processes ensure mutual exclusion. Hence, we can make sure that a chunk
of data is accessed by one process only. 

Boost also provides a function that executes a given function object
\emph{atomically} on a managed shared memory segment such as the one that is
used here to implement the symmetric heap. Hence, we can perform atomic
operations on a (potentially remote) symmetric heap. 

\subsection{Run-time environment}
\label{sec:implem:rte}

As with any parallel program, the run-time environment of OpenSHMEM is here to:
\begin{itemize}[noitemsep]
\item Spawn the parallel processes;
\item Make sure they know how to communicate with each other;
\item Monitor them, and take the appropriate actions if one of them dies;
\item Terminate the execution when necessary;
\item Forward the IOs and signals through the gateway process that provides the
  user with an access to the parallel execution.
\end{itemize}

\paragraph{Process spawning}

Processes are spawned individually by separate threads. At first, a pool of
threads is created: the \texttt{workers} thread group. Then each thread forks a
process: the corresponding OpenSHMEM processing element. The master thread then
yields its slice of time (\texttt{sched\_yield}) and waits on a
condition. Eventually, the threads are joined. 

\paragraph{Contact information}

Processes communicate with each other using shared memory
segments, which are their symmetric shared heap. The name of this
shared memory segment is built using a constant basis and the rank of
the target process. Hence, processes can communicate with each other
as soon as they know their rank. 

\paragraph{Inputs and outputs}

The run-time environment is supposed to forward IOs and signals between the user
and the parallel application. More specifically, the parallel application is
made of several processes, whereas the user is in contact with only one process:
the master process, which is used as a gateway between the user and the
application. 

For instance, if a parallel process performs an output (\texttt{printf},
\texttt{std<<cout}...), the result of this output will be displayed to the user
by the gateway process. Similarly, if the user sends a signal to the gateway
process (\eg \texttt{SIGKILL}), this signal is sent to all the processes of the
parallel application.

The mechanism used here to create the parallel processes preserves IOs. The
parallel processes are offsprings of the gateway processes: hence, their IOs are
forwarded by default.

If necessary, \texttt{stdout}, \texttt{stderr} and \texttt{stdin} can be
duplicated and copied just before the \texttt{execve} system call.

\paragraph{Run-time debugging}

Parallel processes can require to be debugged in an interactive way at
run-time. In this case, a sequential debugger like \texttt{gdb} can be attached
to a given process. 

To allow this attachment, the parallel process is stuck in an infinite loop at
the begining of its initialization. 

Debugging information and checks have to be placed in parts of code that are
removed by the compiler when the debugging mode is disabled. The compiler
variable is called \texttt{\_DEBUG}. 

POSH also provides a \emph{safe} mode. This mode enables some debugging and error
checking information for the parallel program, whereas the debug mode enables
debugging information for the OpenSHMEM library. The compiler
variable is called \texttt{\_SAFE}. 

For instance, the safe mode checks that when a process wants to run a collective
communication, it is not already participating to another collective
communication. 

In order to be able to choose whether to use it or not without affecting the
performance, we chose not to make it a run-time option, but a compile-time
option. As a consequence, it must be enabled by a compiler option when the
OpenSHMEM library is compiled. 

Code related to the safe mode is left out by the compiler when it is not meant
to be enabled. We are using a compiler variable called \texttt{\_SAFE}. 
\section{Performance and experimental results}
\label{sec:perf}

\begin{center}
  \begin{table*}[ht]
    \caption{Comparison of the performance observed with various memcpy implementations\label{tab:perf:memcpy}}
     \small
     \centering

     \begin{subtable}{.48\linewidth}
       \centering
       \begin{tabular}{l|c|c|c|c|}
         ~ & \multicolumn{4}{|c|}{Memory copy latency (ns)} \\
         ~ & memcpy & MMX & MMX2 & SSE \\
         \hline
         \hline
         Caire & 38.85 & 41.10 & 38.65 & {\bf 38.05} \\ 
         \hline
         Jaune &  1277.90  &   1273.90   &  {\bf  1269.90 }  &  1279.90 \\
         \hline
         Magi10 & 45.40   &  {\bf 38.20}   &    39.90   &    40.70  \\
         \hline
         Maximum &  21.70   &  {\bf  20.25}   &    20.45   &    21.00   \\
         \hline
         Pastel & {\bf 1997.30}   & 1997.40 & 2011.35 & 1997.35 \\
       \end{tabular}
     \end{subtable} \hfill %
     \begin{subtable}{.48\linewidth}
       \centering
       \begin{tabular}{l|c|c|c|c|}
         ~ & \multicolumn{4}{|c|}{Memory copy bandwidth (Gb/s)} \\
         ~ & memcpy  & MMX & MMX2 & SSE \\
         \hline
         \hline
         Caire & {\bf 18.40}  &  12.25  &   18.18  &   18.37 \\
         \hline
         Jaune & 9.84   &   10.03  &   16.44 &   {\bf 16.60} \\
         \hline
         Magi10 & {\bf 22.93}  &  21.13  &   17.06 &  20.77 \\
         \hline
         Maximum &  67.47  &   47.52  &   76.59  & {\bf 77.91 }\\
         \hline
         Pastel & 20.27 & 9.12 & {\bf 20.32} & 19.82\\
       \end{tabular} 
     \end{subtable} 
     
  \end{table*} 
\end{center}

This section presents some evaluations of the performance achieved by POSH. Time
measurements were done using {\tt clock\_gettime()} on the {\tt 
  CLOCK\_REALTIME} to achieve nanosecond precision. All the programs were
compiled using {\tt -Ofast} if available, {\tt -O3} otherwise. Each experiment
was repeated 20 times after a warm-up round. We measured the time
taken by data movements (put and get operations) for various buffer
sizes. 

\subsection{Memory copy}
\label{sec:perf:memcpy}

Since memory copy (\texttt{memcpy}-like) is a highly critical function of POSH,
we have implemented several versions of this routine and evaluated them in a separate
micro-benchmark. The compared performance of these various implementations of
{\tt memcpy()} is out of the scope of this paper. A good description of this
challenge, the challenges that are faced and the behavior of the various
possibilities can be found in \cite{memcpySmith}. 

The goal of POSH is to achieve high-performance while being portable. As a
consequence, the choice has been made to provide different implementations of
\texttt{memcpy} and let the user choose one at compile-time, while providing a
default one that achieves reasonably good performance across platforms. 

We have compared several implementations of {\tt memcpy} on various platforms that
feature different CPUs: an Intel Core i7-2600 CPU running at 3.40GHz (Maximum),
a Pentium Dual-Core CPU E5300 running at 2.60GHz (Caire), an AMD Athlon 64 X2
Dual Core Processor 5200+ (Jaune), a large NUMA node featuring 4 CPUs with 10
physical cores each (20 logical cores with hyperthreading), Intel Xeon CPU
E7-4850 running at 2.00GHz (Magi10) and a NUMA node featuring 2 CPUs with 2
cores each, Dual-Core AMD Opteron Processor 2218 running at 2.60GHz
(Pastel). All the platforms are running Linux 3.2, except Jaune (2.6.32) and
Maximum (3.9). The code was compiled by gcc 4.8.2 on Maximum, Caire and Jaune,
gcc 4.7.2 on Pastel and icc 13.1.2 on Magi10. 

We compared the stock {\tt memcpy()} and MMX-, MMX2- and SSE-based
implementations. The performance (latency and bandwidth) of these
implementations on the aforementioned platforms are summarized in
table~\ref{tab:perf:memcpy}. 

We can see that the variations between latencies obtained by all the four
implementations are very small, except for Magi10 (the large NUMA node). Pastel
(the Opteron node) features a slightly better bandwidth with MM2, whereas Jaune
and Maximum (the Athlon XP and the Core i7 nodes) achieve higher bandwidth with
SSE and Caire and Magi10 (the Dual-Core and the large NUMA node) perform better
with the stock {\tt memcpy}. The large performance gap on Jaune may be explained by
the relatively old software it is running. Overall, the stock {\tt memcpy}
performs quite well (best performance of close to the best performance), except
on Jaune and Maximum, on which the bandwidth is largely improved by using SSE
instructions.  

\subsection{POSH communication performance}

\begin{center}
  \begin{figure*}[htb]
    \centering
    \includegraphics[width=.98\linewidth]{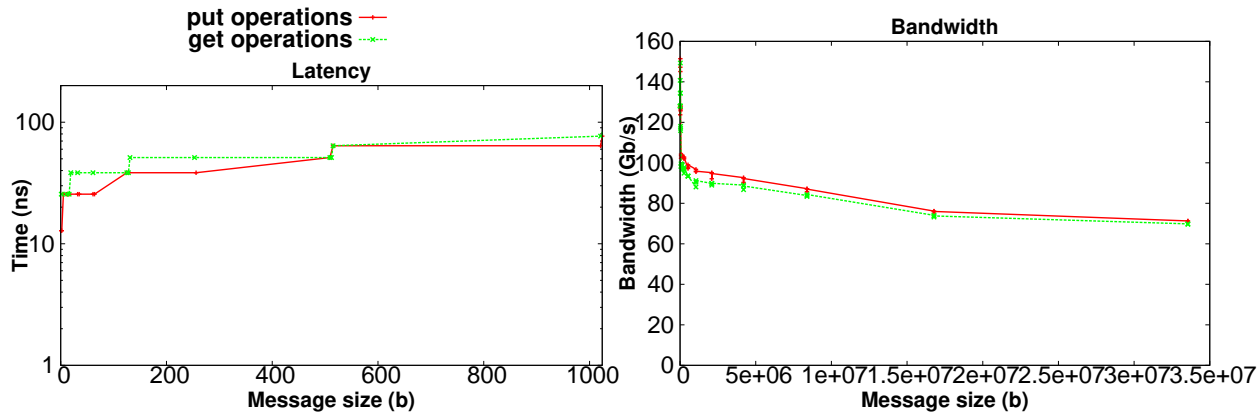}
    \caption{Communication performance of Paris OpenSHMEM on Maximum.\label{fig:perf:posh:maxi}} 
  \end{figure*}
   \begin{table*}[htb]
     \caption{Comparison of the performance observed by put and get operations with POSH\label{tab:perf:posh}}
     \small
     \centering
     \begin{subtable}{.48\linewidth}
       \centering
       \begin{tabular}{l|c|c|c|c|}
         ~ & \multicolumn{4}{|c|}{SHMEM latency (ns)} \\
         ~ & \multicolumn{2}{|c|}{Best copy} & \multicolumn{2}{|c|}{\tt memcpy} \\
         ~ & get & put  & get & put \\
         \hline
         \hline
         Caire & 38.40 & 38.40 & 38.40 & 38.40 \\
         \hline 
         Jaune & 1741.85 & 1665.90 & 1667.90 & 1663.90\\
         \hline
         Magi10  & 38.40 & 38.40 & 38.40 & 38.40 \\
         \hline
         Maximum & 38.40 & 38.40 & 38.40 & 38.40 \\
         \hline
         Pastel & 1830.40 & 1689.60 & 1830.40 & 1689.60 \\
       \end{tabular}
     \end{subtable} %
     \begin{subtable}{.48\linewidth}
       \centering
       \begin{tabular}{l|c|c|c|c|}
         ~ & \multicolumn{4}{|c|}{SHMEM bandwidth (Gb/s)} \\
         ~ & \multicolumn{2}{|c|}{Best copy} & \multicolumn{2}{|c|}{\tt memcpy} \\
         ~ &  get &  put  &  get &  put \\
         \hline
         \hline
         Caire & 18.36 & 18.38 & 18.36 & 18.38 \\
         \hline 
         Jaune & 17.62 & 17.55 & 10.52 & 10.59 \\
         \hline
         Magi10  & 20.46 & 20.16 & 20.46 & 20.16 \\
         \hline
         Maximum & 74.09 & 76.15 & 68.51& 69.28\\
         \hline
         Pastel & 26.07 & 25.50 & 26.07 & 25.50 \\
       \end{tabular}
     \end{subtable} 
   \end{table*} 
 \end{center}

We evaluated the communication performance obtained with POSH. On Caire,
Magi10 and Pastel, we used the stock {\tt memcpy} for data movements. On Jaune
and Maximum, we used both the SSE-based implementation and the sock
{\tt memcpy}. Table \ref{tab:perf:posh} present the latency and
bandwidth obtained by put and get operations with POSH. Figure
\ref{fig:perf:posh:maxi} plots the latency and bandwidth obtained on Maximum.

On the "fast" machines (Caire, Magi10 and Maximum), the latency is too
small to be measured precisely by our measurement method. We can
see on table \ref{tab:perf:posh} that the latency has the same order of magnitude as the
one obtained by a {\tt memcpy} within the memory space of a single
process. However, measuring the latency on regular communication
patterns gives an indication on the overall latency of the
communication engine, but may be different from what would be obtained
on more irregular patterns, where the segment of shared memory is not
in the same cache memory as the process that performs the data
movement. In the latter case, the kernel's scheduling performance is
highly critical.

Similarly, the bandwidth obtained by POSH has little overhead compared
with the one obtained by a {\tt memcpy} within the memory space of a single
process. We can conclude here that our peer-to-peer communication
engine adds little overhead, no to say a negligible one, and
inter-process communications are almost as fast as local memory copy
operations. 

\subsection{Comparison with another communication library}

\begin{center}
  \begin{table*}[ht]
    \caption{Comparison of the performance observed by put and get operations
      with UPC\label{tab:perf:upc}}
     \small
     \centering
     
     \begin{subtable}{.48\linewidth}
       \centering
       \begin{tabular}{l|c|c|}
         ~ & \multicolumn{2}{|c|}{UPC latency (ns)} \\
         ~ &  get & put \\
         \hline
         \hline
         Caire & 39.40 & 37.55 \\
         \hline
         Jaune & 1623.90 & 1623.90 \\
         \hline
         Magi10 & 73.80   & 54.90 \\
         \hline
         Maximum & 26.75   &  25.00 \\
         \hline
         Pastel & 2025.10 &  1689.95\\
     \end{tabular}\end{subtable} %
     \begin{subtable}{.48\linewidth}
       \centering
       \begin{tabular}{l|c|c|}
         ~ & \multicolumn{2}{|c|}{UPC bandwidth (Gb/s)} \\
         ~ &  get & put \\
         \hline
         \hline
         Caire & 18.03 & 18.45 \\
         \hline
         Jaune & 9.95 & 10.63 \\
         \hline
         Magi10 & 18.64 & 16.33 \\
         \hline
         Maximum & 67.45 & 68.86 \\
         \hline
         Pastel & 23.52 & 25.06 \\
       \end{tabular}
     \end{subtable} 
  \end{table*} 
\end{center}

We used a similar benchmark to evaluate the communication performance
of Berkeley UPC, whose communication engine, GASNet \cite{GASNET},
uses {\tt memcpy} to move data. As a consequence, the results obtained
here must be compared to those obtained in the previous sections with
the stock {\tt memcpy}. Here again, we can see that BUPC inter-process
data movement operations have little overhead compared to a memory
copy that would be performed within the memory space of a single
process. The results are presented in table \ref{tab:perf:upc}.

We can see here that both POSH and another one-sided communication library (Berkeley
UPC) have performance that are close to a memory copy within the address space
of a single process. Besides, we have seen how the performance can benefit from
a tuned memory copy routine. 
\section{Conclusion and perspective}
\label{sec:conclu}

In this paper, we have presented the design and implementation of POSH, an
OpenSHMEM implementation based on a shared memory engine provided by
Boost.Interprocess, which is itself based on the POSIX {\tt shm} API. We have
presented its architecture, a model for its communications and proved some
properties that some implementation choices rely on. We have presented an
evaluation of its performance and compared it with a state-of-the-art
implementation of UPC, another programming API that follows the same
communication model (one-sided communications).

We have seen that POSH achieves a performance which is comparable with this
other library and with simple memory copies. We have also shown how it can be
tuned in order to benefit from optimized low-level routines. 

That communication model opens perspectives on novel work on distributed
algorithms. The architectural choices that were made in POSH make it possible to
use it as a platform for implementing and evaluating them in practice. For
instance, locks, atomic operations and collective operations are classical
problems in distributed algorithms. They can be reviewed and
re-examined in that model in order to create novel algorithms with an
original approach.

\bibliographystyle{plain}
\bibliography{iccs}

\end{document}